\documentclass[12pt,reqno]{article}

\usepackage{xcolor} 
\usepackage{amssymb}
\usepackage{amsmath}
\usepackage{amsthm}
\usepackage{amsfonts}
\usepackage{amscd}
\usepackage{graphicx}

\usepackage{float} 

\usepackage[colorlinks=true,
linkcolor=webgreen,
filecolor=webbrown,
citecolor=webgreen]{hyperref}

\definecolor{webgreen}{rgb}{0,.5,0}
\definecolor{webbrown}{rgb}{.6,0,0}

\usepackage{color}
\usepackage{fullpage}

\usepackage{graphics}
\usepackage{latexsym}
\usepackage{epsf}

\newcommand{\seqnum}[1]{\href{https://oeis.org/#1}{\rm \underline{#1}}}

\usepackage{pgfplots}
\pgfplotsset{compat=1.18}
\usepackage{tikz-cd}
\usepackage{enumitem}

\newcommand{\red}[1]{{\color{red} #1}}

\newcommand{\freq}{\mathrm{freq}}

\begin{document}

\vspace*{1cm} 


\theoremstyle{plain}
\newtheorem{theorem}{Theorem}
\newtheorem{corollary}[theorem]{Corollary}
\newtheorem{lemma}[theorem]{Lemma}
\newtheorem{proposition}[theorem]{Proposition}

\theoremstyle{definition}
\newtheorem{definition}[theorem]{Definition}
\newtheorem{example}[theorem]{Example}
\newtheorem{conjecture}[theorem]{Conjecture}

\theoremstyle{remark}
\newtheorem{remark}[theorem]{Remark}

\numberwithin{equation}{section}
\numberwithin{theorem}{section}
\numberwithin{figure}{section}

\begin{center}
\vskip 1cm{\LARGE\bf
Pointwise Order of Generalized Hofstadter Functions $G$, $H$, and Beyond\\
}
\vskip 1cm
\large
    Pierre Letouzey\\
    IRIF \\
    Universit\'e Paris Cit\'e, CNRS, Inria\\
    F-75013 Paris\\
    France\\
    \href{mailto: letouzey@irif.fr}{\tt letouzey@irif.fr}\\
    \ \\
    Shuo Li\\
    Hangzhou International Innovation Institute \\
    Beihang University \\
    166 Shuanghongqiao Street \\
    Pingyao Town, Yuhang District, Hangzhou, 311115\\
    China \\
    \href{mailto: shuoli@buaa.edu.cn}{\tt shuoli@buaa.edu.cn}\\
    \ \\
    Wolfgang Steiner\\
    IRIF\\
    Universit\'e Paris Cit\'e, CNRS\\
    F-75013 Paris\\
    France\\
    \href{mailto: steiner@irif.fr}{\tt steiner@irif.fr}\\
\end{center}

\vskip .2 in
\begin{abstract}
    Hofstadter's G function is recursively defined via $G(0)=0$ and then $G(n)=n-G(G(n-1))$. Following Hofstadter, we vary the number $k$ of nested recursive calls in this equation and obtain a family of functions $(F_k)$.
    Here we establish that this family is ordered pointwise: for all $k$ and $n$, we have $F_k(n) \le F_{k+1}(n)$. To achieve this, we make a detour via infinite morphic words generalizing the Fibonacci word. We prove various properties of these words, concerning the lengths of substituted prefixes of these words and the number of occurrences of specific letters in these prefixes. We also relate the limits of $\frac{1}{n}F_k(n)$ to the frequencies of letters in the considered words. We provide a certified formalization of all these results in the Rocq proof assistant.
\end{abstract}

\section{Introduction}

\subsection{The functions}

For each integer $k \ge 1$, we recursively define the
function\footnote{In all this work, the
$k$-th exponent of a function denotes its $k$-th iterate.
For instance, $F_2^2(n)$ is $F_2(F_2(n))$, not the square of $F_2(n)$.}
\[
F_k:\, \mathbb{N} \to \mathbb{N}, \quad n \mapsto \begin{cases}0, & \mbox{if}\ n = 0; \\ n-F_k^k(n-1), & \mbox{otherwise}.\end{cases}
\]
The function $F_k$ is well defined, since one may prove that
$0 \le F_k(n) \le n$ for all $n \ge 0$. This family of functions is
due to Hofstadter \cite[Chap.~5]{GEB}.
In particular, $F_2$ is Hofstadter's
function~$G$ (see OEIS entry~\seqnum{A005206}
\cite{OEIS,GEB,DowneyGriswold84,GaultClint}). Note that
$G(n)=\lfloor (n+1)/\varphi \rfloor$ where $\varphi$ is the
golden ratio $(1+\sqrt{5})/2$. Similarly, $F_3$ is
Hofstadter's function~$H$ (see OEIS \seqnum{A005374}),
 and the generalization to
higher degrees of recursive nesting in the definition of~$F_k$
was already suggested by Hofstadter~\cite{GEB}.
To be complete, the OEIS database already includes
$F_4$ as \seqnum{A005375},
$F_5$ as \seqnum{A005376}, and
$F_6$ as \seqnum{A100721}. On the other hand, we
choose to start this sequence with~$F_1$, where only one recursive call
is done, leading to a function that can easily be shown to verify
$F_1(n)=\lfloor (n+1)/2 \rfloor=\lceil n/2 \rceil$.
Throughout this paper, we never consider the case $k=0$:
although the previous definition could be extended and gives a
non-recursive function~$F_0$, this $F_0$
has too little in common with the other~$F_k$ functions to be
of much interest.

\begin{figure}[ht]
\begin{center}
\pgfplotsset{width=\linewidth}
\begin{tikzpicture}[scale=.9]
  \begin{axis}[
    xmin=0,xmax=30,ymin=0,ymax=25,samples=31,
    xtick = {0,5,...,30},
    ytick = {0,5,...,30},
    legend pos=south east
  ]
 \addplot+[mark=.,color=black,style=dashed,domain=0:30] {x};
 \addlegendentry{id \phantom{=G}}
    \addplot [mark=.,color=blue] coordinates {
(0, 0) (1, 1) (2, 1) (3, 2) (4, 3) (5, 4) (6, 5) (7, 6) (8, 6)
 (9, 7) (10, 7) (11, 8) (12, 9) (13, 9) (14, 10) (15, 11) (16, 12)
 (17, 12) (18, 13) (19, 14) (20, 15) (21, 16) (22, 16) (23, 17)
 (24, 18) (25, 19) (26, 20) (27, 21) (28, 21) (29, 22) (30, 23)};
 \addlegendentry{$F_5 \phantom{=G}$}
    \addplot [mark=.,color=purple] coordinates {
(0, 0) (1, 1) (2, 1) (3, 2) (4, 3) (5, 4) (6, 5) (7, 5) (8, 6)
 (9, 6) (10, 7) (11, 8) (12, 8) (13, 9) (14, 10) (15, 11) (16, 11)
 (17, 12) (18, 13) (19, 14) (20, 15) (21, 15) (22, 16) (23, 17)
 (24, 18) (25, 19) (26, 19) (27, 20) (28, 20) (29, 21) (30, 22)};
 \addlegendentry{$F_4 \phantom{=G}$}
    \addplot [mark=.,color=orange] coordinates {
(0, 0) (1, 1) (2, 1) (3, 2) (4, 3) (5, 4) (6, 4) (7, 5) (8, 5)
 (9, 6) (10, 7) (11, 7) (12, 8) (13, 9) (14, 10) (15, 10) (16, 11)
 (17, 12) (18, 13) (19, 13) (20, 14) (21, 14) (22, 15) (23, 16)
 (24, 17) (25, 17) (26, 18) (27, 18) (28, 19) (29, 20) (30, 20) };
 \addlegendentry{$F_3=H$}
 \addplot+[mark=.,color=red,style=solid,domain=0:30]
 {floor((x+1) * 0.618033988749894903)};
 \addlegendentry{$F_2=G$}
 \addplot+[mark=.,color=black,style=solid,domain=0:30] {ceil(x/2)};
 \addlegendentry{$F_1 \phantom{=G}$}
 \end{axis}
\end{tikzpicture}
\end{center}
\caption{Plotting $F_1,F_2,\ldots,F_5$}
\label{fig:plot}
\end{figure}

\subsection{A monotonicity problem}

Small values of the functions $F_1$ to~$F_5$ are displayed in
Figure~\ref{fig:plot}. From this figure, one may easily guess
that $F_1$ is everywhere below or equal to~$F_2$,
similarly that $F_2$ is below or equal to~$F_3$, and so forth. Indeed, our main result is the following theorem.

\begin{theorem} \label{t:2}
For all $k \ge 1$, $n \ge 0$, we have $F_k(n) \le F_{k+1}(n)$.
\end{theorem}

This property
seems to defy any attempt to prove it directly via induction on the functions~$F_k$. Therefore, a~different approach is used.
We relate each function~$F_k$ to a morphic word~$x_k$ and turn the
function comparison into an equivalent statement about these words.
This gives Theorem~\ref{t:2a} below, which contains Theorem~\ref{t:2} as a special case.

\subsection{Substitutions and morphic words}
\label{subs:substitution}

For $k \ge 1$, let $\tau_k$ be the substitution
 on the alphabet $\{1,2,\dots,k\}$ defined by 
\[
\begin{aligned}
\tau_k:\, k &\mapsto k\,1, \\
i & \mapsto i+1, \quad \mbox{for}\ 1 \le i < k,
\end{aligned}
\]
and let $x_k = x_k[0] x_k[1] \cdots \in \{1,2,\dots,k\}^\infty$ be the
fixed point of~$\tau_k$.
Here, a~substitution (or morphism) on an alphabet $A$ is a map $\tau: A^* \to A^*$ satisfying $\tau(uv) = \tau(u) \tau(v)$ for all $u,v \in A^*$, where $A^*$ denotes the set of finite words with letters in~$A$.
The map $\tau$ is therefore defined by its value on the letters of~$A$, and it is extended in a natural way to infinite words (or sequences) $w = w[0] w[1] \cdots \in A^\infty$ by setting $\tau(x) = \tau(x[0]) \tau(x[1]) \cdots$.
For more on substitutions, we refer to Pytheas Fogg~\cite{Pytheas}.

\begin{figure}[ht]
\begin{align*}
x_2 &= \red{2}1\red{2}\red{2}1\red{2}1\red{2}\red{2}1\red{2}\red{2}1\red{2}1\red{2}\red{2}1\red{2}1\red{2}\red{2}1\red{2}\red{2}1\red{2}1\red{2}\red{2}1\red{2}\red{2}1\red{2}1\red{2}\red{2}1\red{2}\cdots \\
x_3 &= \red{3}12\red{3}\red{3}1\red{3}12\red{3}12\red{3}\red{3}12\red{3}\red{3}1\red{3}12\red{3}\red{3}1\red{3}12\red{3}12\red{3}\red{3}1\red{3}12\red{3}12\cdots \\
x_4 &= \red{4}123\red{4}\red{4}1\red{4}12\red{4}123\red{4}123\red{4}\red{4}123\red{4}\red{4}1\red{4}123\red{4}\red{4}1\red{4}12\red{4}123\cdots \\
x_5 &= \red{5}1234\red{5}\red{5}1\red{5}12\red{5}123\red{5}1234\red{5}1234\red{5}\red{5}1234\red{5}\red{5}1\red{5}1234\red{5}\cdots
\end{align*}
\caption{Infinite words $x_2,\ldots,x_5$ (with largest letter in red)}
\label{f:xk}
\end{figure}

These substitutions and words are not novel.
For instance, $\tau_1$ operates on the singleton alphabet~$\{1\}$ in such a way that $\tau_1(1)=11$, and hence $x_1 = 1^\infty$. It is the
only case where $x_k$ is ultimately periodic.
For $k=2$, we recover the well-known Fibonacci substitution and word.\footnote{Usually, the letter~$k$ is replaced by the letter~$0$ in the definition of $\tau_k$ and~$x_k$. We choose the letter~$k$ in order to simplify some formulas and to extend the definition without modification to the case $k=1$.}
The substitution~$\tau_3$ already appears (up to letter renaming) as $\sigma(1,0)$ in Pytheas Fogg~\cite[Ex.~8.1.2]{Pytheas}, as an example of a~\emph{modified Jacobi--Perron substitution}.
More generally, the substitutions~$\tau_k$ and words~$x_k$ can be associated with the R\'enyi
expansion of~1 in base~$\beta_k$, where $\beta_k$ is the positive zero of
$X^k - X^{k-1} - 1$, and were studied by Frougny, Mas\'akov\'a and Pelantov\'a~\cite{FrougnyMP04}. 
When $k\ge 2$, this expansion can be written as $1(0)^{k-2}1$. As a consequence,
the factor complexity of~$x_k$ is $n \mapsto (k-1) n+1$; in particular, $x_2$ is Sturmian \cite{FrougnyMP04}.

\subsection{New results}

In this article, in addition to the aforementioned monotonicity of~$F_k$ (Theorem~\ref{t:2}), we exhibit several relations between the
functions~$F_k$ and the words~$x_k$.
For instance, for $k\ge 1$ and a position $n\ge 0$,
the value of $x_k[n]$ is either the first index $j\le k$
such that $F_k^j(n+1) = F_k^j(n)$, or $k$ if there is no such small
index; see Proposition~\ref{diffFx} below.
In particular, $x_k[n]$ is the letter~$1$ exactly when
$F_k(n+1) = F_k(n)$, and this also leads to the fact that
$F_k(n)$ is equal to the number of non-$1$ letters in the
first $n$ letters of~$x_k$; see Proposition~\ref{FC1}.
Such properties were already known in the case $k=2$, where
$F_2 = G = \lfloor (n+1)/\varphi \rfloor$ and $x_2$ amounts to the Fibonacci
word; see, for instance, Allouche and Shallit \cite[Ex.~8.1.3]{allouche_shallit_2003}.
The extended results to an arbitrary parameter~$k$ seem new.

All these relations between the functions~$F_k$ and the words~$x_k$
stem from the key Theorem~\ref{t:1} and its Corollary~\ref{galois}:
$F_k$~and their iterates admit Galois connections~\cite{Erne1993} (i.e., almost
inverses) that can be simply expressed in terms of the length of
$\tau_k$-substituted prefixes of~$x_k$. These lengths are called~$L_k$
here; see Proposition~\ref{basicL}.

The words~$x_k$ are primarily considered here in
order to study the monotonicity of the functions~$F_k$. Nonetheless, several
results on~$x_k$ are also of interest on their own. For example,
for $k\ge 1$ and $n\ge 0$, we prove that there are more occurrences of the letter~$k$
among the first $n$ letters of~$x_k$ than occurrences of the letter $k+1$ among
the first $n$ letters of~$x_{k+1}$ (Theorem~\ref{p:1}). Similarly,
the letter~$1$ occurs more often among the first~$n$ letters of~$x_k$ than among the first~$n$ letters of~$x_{k+1}$ (Proposition~\ref{p:C1decr}). We denote the number of letters~$i$ among the prefix of~$x_k$ of length~$n$ by~$C_k^{(=i)}(n)$;
see Section~\ref{s:count}.

\subsection{The Rocq formalization}

All the results presented in this article have been formally certified
using the Rocq proof assistant~\cite{Rocq} (previously known as
Coq). This implementation in Rocq is freely available~\cite{LetouzeyCoqDev}.
Thanks to this Rocq formalization, we ensure precise
definitions and statements while preventing any error during
the corresponding proofs. It can hence serve as a reference for the
interested reader.
In particular, the recommended entry point to read alongside this article~is
\begin{center}
\url{https://github.com/letouzey/hofstadter_g/blob/main/Article1.v}
\end{center}
The current article tries to be faithful to this formal work
while staying readable by a large audience, at the cost of possible
remaining mistakes during the transcription.

As a quick illustration, here is one possible Rocq encoding of the $F_k$
function:\footnote{See \texttt{Definition~f} in the file \texttt{GenG.v}~\cite{LetouzeyCoqDev}.
Faster implementations of~$F_k$ are also provided in the file \texttt{Fast.v} and proved equal with the version shown here.}
\begin{verbatim}
Fixpoint recf k p n : nat := match p with
| 0 => 0
| S p' => n - Nat.iter k (recf k p') (n-1)
end.
Definition f k n := recf k n n.
\end{verbatim}
Note that \texttt{Nat.iter~k} iterates a given function \texttt{k} times.
This encoding departs slightly from the previous definition of~$F_k$:
in order to satisfy the Rocq constraints on recursive definitions,
an extra parameter~\texttt{p} is initialized to~\texttt{n} and
decreased in recursive calls. Then we proved that this Rocq definition
fulfills the desired equation:\footnote{See \texttt{Lemma~f\_eqn} and its proof in the file \texttt{GenG.v}~\cite{LetouzeyCoqDev}.
Note that the Rocq subtraction on natural numbers of type~\texttt{nat} is bounded by~\texttt{0}, for example \texttt{0-1~=~0}. Therefore, the statement of~\texttt{f\_eqn} holds both when $\mathtt{n}=0$ and when $\mathtt{n}\neq 0$.}
\begin{verbatim}
Lemma f_eqn : forall k n : nat, f k n = n - Nat.iter k (f k) (n-1).
\end{verbatim}
In this setting, our main Theorem~\ref{t:2} has a very simple Rocq statement:\footnote{
See \texttt{Theorem~f\_grows} and its proof in the file \texttt{WordGrowth.v}~\cite{LetouzeyCoqDev}.}
\begin{verbatim}
Theorem f_grows : forall k n, f k n <= f (k+1) n.
\end{verbatim}

The provided Rocq files can be machine-checked again by any recent
installation of Rocq; see the joint \texttt{README} file in the
repository~\cite{LetouzeyCoqDev}.
Since Sections~\ref{s:poly} and~\ref{s:infini} of this article involve real numbers, their corresponding Rocq files rely on
the Rocq library \texttt{Reals} and its four logical axioms,
in particular the axiom of excluded middle. All the rest of the
formalization, including our key Theorem~\ref{t:2}, has been performed within
Rocq core logic, with no axiom. This fact can be checked
via the Rocq command \texttt{Print Assumptions} on our theorems.

\subsection{Summary}
After introducing some notation and basic properties in Section~\ref{s:notation},
we establish a key link between the functions~$F_k$ and words~$x_k$ in Section~\ref{s:relating_F_L}.
In Section~\ref{s:count}, we count some letters in prefixes of~$x_k$.
Section~\ref{s:poly} presents some zeros of polynomials that are
used in Section~\ref{s:infini} when stating the infinitary behavior
of our objects, i.e., when $n \to \infty$.
In particular, we present linear equivalents of~$F_k$ and a weak
form of monotonicity for the family~$(F_k)$.
Section~\ref{s:kmono} contains the promised full monotonicity proof
for all $n\ge 0$,
at first for words, then functions, and various extensions are made or
conjectured.
Section~\ref{s:morecount} revisits the previous results, this time
in terms of the number of occurrences of letters.

\subsection{Related works}
Various aspects of the functions~$F_k$ can be studied.
We mention some of them, even if they do not seem to be connected to the monotonicity, which is the main focus of this article. 
 
First, the case $k=2$ has been well explored
\cite{DowneyGriswold84,GaultClint}.
However, the fact that $F_2$ has an exact expression
$\lfloor (n+1)/\varphi \rfloor$ is quite particular and too specific
to draw general conclusions about all the other~$F_k$.

A~nice general fact about the functions~$F_k$ is that they can be
presented as right shifts of digits when numbers~$n$ are written
with the appropriate numerical representations. In particular, for
$k=2$, the corresponding numerical representation is the Zeckendorf
decomposition, writing numbers as sums of distinct Fibonacci numbers.
This can be extended to other values of~$k$ by changing the sequence
used as the base for the decomposition, from Fibonacci numbers to a similar linear recurrence \cite{MeekVanRees84,DowneyGriswold84,Kimberling95,EriksenAnderson12,Dekking2023,Shallit2024}.
In particular, Meek and Van Rees~\cite{MeekVanRees84} focus on a
slightly different family of functions, which can be written here as
$F_k(n+1) - 1$.
If we denote these ``shifted'' functions by~$\tilde{F}_k$, they verify the
equations $\tilde{F}_k(0)=0$ and $\tilde{F}_k(n) = n-1 - \tilde{F}_k^k(n-1)$; note the extra~$-1$ when
compared to the equation of~$F_k$.
In any case, having a description of~$F_k(n)$ (or~$\tilde{F}_k(n)$) in terms
of a $k$-decomposition of~$n$ does not directly help for our
monotonicity problem, since the $k$- and
$(k+1)$-decompositions of $n$ differ too much, just as a number tends to have very different base-2 and base-3 digits.
The concept of $k$-decompositions is also related to
the notions of substitutions and words that we exploit here directly.

As a final remark about~$F_k$, note that $F_k$ can also be presented
as a companion infinite tree where the same pattern is continuously
repeated: a binary node is followed on its right by $k-1$ unary
nodes. When the nodes of this tree are labelled by increasing numbers,
in a breadth-first manner, left to right, then the node number~$n$ has
$F_k(n)$ as parent node. This was already described by
Hofstadter~\cite{GEB}; also see the general introduction to the problem by the first author~\cite{Letouzey15}. Here again,
these alternative presentations of~$F_k$ do not seem to help
concerning the monotonicity, with no direct relations between
trees for~$k$ and for~$k+1$.

The sequences~$x_k$ defined in Section~\ref{subs:substitution} can be considered as generalizations of the Fibonacci word. However, one can find many different ``generalized Fibonacci'' words in the literature, for example the $k$-bonacci words, which can be defined as fixed points of the morphism satisfying $1 \mapsto k$, $i \mapsto k(i-1)$, for $2 \le i \le k$ \cite{Tanwen07}, the fixed points of the morphisms $0 \mapsto 0^m1^n$, $1 \mapsto 0$, for positive integers $m,n$ \cite[Ex.~10.11.18]{allouche_shallit_2003}, and other examples \cite{RAMIREZ2014,GOLDBERG2018,HARE2022}.

\section{Notation and basic properties}
\label{s:notation}

We let $\partial F_k^j(n)$ denote the difference
$F_k^j(n+1) - F_k^j(n)$.

\begin{proposition}
\label{basicF}
For all $j,k \ge 1$, the function $F_k^j$ satisfies the following
basic properties:
\begin{enumerate}[label=(\alph*)]
\itemsep.5ex
\item $F_k^j(1) = F_k^j(2) = 1$,
\item $F_k^j(n) \ge 1$ whenever $n\ge 1$,
\item $F_k^j(n) < n$ whenever $n\ge 2$,
\item we have $\partial F_k(n) = 1 - \partial F_k^k(n-1)$ for all $n \ge 1$, 
\item hence $\partial F_k^j(n) \in \{0,1\}$ for all $n\ge 0$,
\item the function $F_k^j$ is monotonically increasing and onto (but not one-to-one).
\end{enumerate}
\end{proposition}
\begin{proof}
Direct use of the definition or easy induction on~$n$.
\end{proof}

For $n\ge 0$, let $x_k[0{:}n) = x_k[0]x_k[1]\cdots x_k[n-1]$ denote
the prefix of~$x_k$ of length~$n$. In particular, $x_k[0{:}0)$ is
the empty word, and $\big|x_k[0{:}n)\big| = n$, where $|{\cdot}|$ denotes the length of~a~word.

Note that $x_k$ can also be seen as the fixed point of the substitution~$\tau_k^{k-1}$.
Since $\tau_k^{k-1}(i) = k12\cdots(i-1)$ for all $1 \le i \le k$,
this provides a decomposition of~$x_k$ into blocks delimited by the
letter~$k$ (in red in Figure~\ref{f:xk}). The lengths of these blocks
are the successive letters of~$x_k$, and this gives us a convenient
way to compute~$x_k$. For instance,
$\tau_3^2(3)=\red{3}12$, then
$\tau_3^{2}(312)=\red{3}12\,\red{3}\,\red{3}1$, etc.

Similarly, $x_k$ is also the fixed point of~$\tau_k^k$.
Since $\tau_k^k(i) = k12\cdots i$ for all $1 \le i \le k$, this
other decomposition is interesting for counting occurrences of the letter~$1$
in~$x_k$, because $1$ occurs only in the second position of each such
block.

As an alternative way of computing~$x_k$, note that its prefixes of
the form~$\tau_k^j(k)$ follow the following base cases and recursive rule.

\begin{proposition}
\label{p:taujk_rec}
Let $k\ge 1$. Then
\[
\tau_k^j(k) = \begin{cases} k12\cdots j, & \text{if}\ 0\le j\le k;\\
\tau_k^{j-1}(k)\,\tau_k^{j-k}(k), & \text{if}\ j\ge k. \end{cases}
\]
\end{proposition}

\begin{proof}
We use induction on~$j$. For $j=0$, we have $\tau_k^0(k)=k$.
Otherwise, $\tau_k^j(k) = \tau_k(\tau_k^{j-1}(k))$, and we use the induction
hypothesis on $j-1$. If $j\le k$, then
\[
\tau_k^j(k) = \tau_k(\tau_k^{j-1}(k)) = \tau_k(k12\cdots (j-1)) =
k12\cdots j,
\]
which also gives the recursive rule for $j=k$ since
$k12\cdots(k-1)k = \tau_k^{k-1}(k)\tau_k^0(k)$.
If $j>k$, then $j-1\ge k$, and
\[
\tau_k^j(k) = \tau_k(\tau_k^{j-1}(k)) =
\tau_k(\tau_k^{j-1-1}(k)\tau_k^{j-1-k}(k)) =
\tau_k^{j-1}(k)\tau_k^{j-k}(k). \qedhere
\]
\end{proof}

For $k\ge 1$, an important counterpart to~$F_k$ is the function
\[
L_k:\, \mathbb{N} \to \mathbb{N}, \quad n \mapsto \big|\tau_k(x_k[0{:}n))\big|.
\]

\begin{proposition}
\label{basicL}
The lengths $L_k(n) = \big|\tau_k(x_k[0{:}n))\big|$
satisfy the following basic properties for all $k\ge 1$ and $j,n \ge 0$:
\begin{enumerate}[label=(\alph*)]
\itemsep.5ex
\item the $j$-th iterate of $L_k$ satisfies
$L_k^j(n) = \big|\tau_k^j(x_k[0{:}n))\big|$,
\item $L_k^j(0)=0$ and $L_k^j(n+1)=L_k^j(n)+\big|\tau_k^j(x_k[n])\big|$,
\item $L_k^j(1)=j+1$ when $j\le k$ and
  $L_k^j(1)=L_k^{j-1}(1)+L_k^{j-k}(1)$ when $j\ge k$,
\item the function $L_k^j$ is strictly monotonically increasing over $n$,
\item $L_k^j(n) \ge n$, with equality only when $j=0$ or $n=0$,
\item for $n \ge 1$, the function $j\mapsto L_k^j(n)$ is strictly monotonically increasing over $j$.
\end{enumerate}
\end{proposition}

\begin{proof}
First, $\tau_k(x_k[0{:}n))$ is the prefix of~$x_k$ of length~$L_k(n)$. It is
  hence equal to $x_k[0{:}L_k(n))$. 
We claim that, more generally, $\tau_k^j(x_k[0{:}n))$ is the prefix of~$x_k$ of length~$L_k^j(n)$ for all $j \ge 0$. 
Since $\tau_k^j(x_k[0{:}n))$ is a prefix of~$x_k$, we only have to prove that its length is equal to $L_k^j(n)$, and this follows inductively for $j \ge 1$ from
\[
\big|\tau_k^j(x_k[0{:}n))\big|=\big|\tau_k(\tau_k^{j-1}(x_k[0{:}n)))\big|
=\big|\tau_k(x_k[0{:}L_k^{j-1}(n)))\big| = L_k(L_k^{j-1}(n)) = L_k^j(n).
\]
Point (b) is a direct consequence of this more general claim.
Point (c) directly comes from Proposition~\ref{p:taujk_rec}.

Now a key fact: for all finite words $w$, we have $|w| \le |\tau_k(w)|$.
This inequality is even strict when $w$ is a non-empty prefix of $x_k$,
since $x_k[0]=k$ and $|\tau_k(k)|=2$. The remaining points are direct
consequences of this key fact.
\end{proof}

\begin{proposition} \label{p:bound_unique}
For $j\ge 0$ and $k,m\ge 1$, there exists a unique $n\ge 1$
such that $L_k^j(n-1)<m\le L_k^j(n)$.
\end{proposition}
\begin{proof}
Standard consequence of $L_k^j(0)=0$ and the strict monotonicity of $L_k^j$.
\end{proof}

\section{Relating functions and word lengths}\label{s:relating_F_L}

We now establish that~$L_k^j$ allows expressing preimages of the
function~$F_k^j$.

\begin{theorem} \label{t:1}
Let $j\ge 0$ and $k\ge 1$.
\begin{enumerate}[label=(\alph*)]
\item
For all $n \ge 1$, we have
\[
F_k^{-j}(\{n\}) = \big(L_k^j(n-1),L_k^j(n)\big] \cap \mathbb{N}.
\]
\item
Equivalently, for all $m \ge 1$, we have
\begin{equation} \label{e:Fkj}
L_k^j\big(F_k^j(m) - 1\big) < m \le L_k^j\big(F_k^j(m)\big).
\end{equation}
\end{enumerate}
\end{theorem}

\begin{proof}
We first prove the equivalence of (a) and~(b).
The case~(a) directly implies~(b); just replace $n$ by~$F_k^j(m)$. 
Now we prove the set equation in~(a) assuming~(b):
If $m \in F_k^{-j}(\{n\})$, then $F_k^j(m)=n$, and \eqref{e:Fkj} becomes
$L_k^j(n-1)<m\le L_k^j(n)$. Hence $m$ is a natural number in the
right interval. On the other hand, if
$m \in \big(L_k^j(n-1),L_k^j(n)\big] \cap \mathbb{N}$, then Proposition~\ref{p:bound_unique} and \eqref{e:Fkj} imply that $F_k^j(m)=n$,
so $m\in F_k^{-j}(\{n\})$. The two sets in~(a) are hence equal.

We now prove the statement~(b).
The case $k=1$ must be separately handled, since our general proof
below requires the letters $1$ and $k$ to differ.
Fortunately, when $k=1$ we have $L_1^j(n) = 2^{j}n$ and
$F_1^j(m)=\lceil m/2^j \rceil$, which easily resolves this case.
We can also separately handle
the case $j=0$, which is obvious since both $L_k^0$ and~$F_k^0$ are the
identity functions.

From now on, we assume $k \ge 2$ and prove the inequalities \eqref{e:Fkj} for all $j \ge 1$ by strong induction on
$m \ge 1$.
For $m \in \{1,2\}$ and $j \ge 1$, we have $F_k^j(m) = 1$ by Proposition~\ref{basicF}, 
while $L_k^j(0)=0$ and $L_k^j(1) > L_k^0(1) = 1$ by Proposition~\ref{basicL}. Hence the desired inequalities hold in this case.

Let $n \ge 2$, and assume that \eqref{e:Fkj} holds for all $1 \le m \le
n$, $j \ge 1$. We now prove the inequalities in \eqref{e:Fkj} for $m = n+1$, first in the case $j=1$,
and then for $j\ge 2$.

\smallskip\noindent
\textbf{Case} $j = 1$.
We use the following instances of the induction hypothesis:
\begin{align}
\label{IHn} L_k(F_k(n) - 1) &< n \le L_k(F_k(n)), \\
\label{IHnk} L_k^k(F_k^k(n) - 1) &< n \le L_k^k(F_k^k(n)), \\
\label{IHpred} L_k^k(F_k^k(n-1) - 1) &< n-1 \le L_k^k(F_k^k(n-1)).
\end{align}
Also recall from Proposition~\ref{basicF} that
\[
F_k(n+1) - F_k(n) = 1 - \big(F_k^k(n) - F_k^k(n-1)\big)
\]
and $F_k^k(n) - F_k^k(n-1)$ is either $0$ or $1$.
We separately handle these two possible differences.

\smallskip\noindent
\textbf{Case} $j = 1$, $F_k^k(n) = F_k^k(n-1)$.
If $F_k^k(n) = F_k^k(n-1)$, then $F_k(n+1) = F_k(n)+1$.
From~\eqref{IHn} and the strict monotonicity of $L_k$
(Proposition~\ref{basicL}), we get
\[
n \le L_k(F_k(n)) < L_k(F_k(n)+1) = L_k(F_k(n+1)).
\]
Hence the right inequality in \eqref{e:Fkj} holds for $m = n+1$, $j=1$.
We claim that $L_k(F_k(n)) = n$.
If this were not the case, then \eqref{IHn} would give
$L_k(F_k(n) - 1) \le n-1 < n < L_k(F_k(n))$.
For the words corresponding to these lengths, this would imply that
$x_k[n-1]x_k[n]$ is a subword of $\tau_k(x_k[F_k(n) - 1])$, and in particular that $|\tau_k(x_k[F_k(n) - 1])| \ge 2$. This could only happen when $x_k[F_k(n) - 1] = k$, and thus $x_k[n]=1$.
From $F_k^k(n) = F_k^k(n-1)$, \eqref{IHnk} and \eqref{IHpred}, we get $L_k^k(F_k^k(n) - 1) \le n-2 < n-1 < L_k^k(F_k^k(n))$, and hence $x_k[n-2]x_k[n-1]$ is a subword of the block $\tau_k^k(x_k[F_k^k(n) - 1])$.
Recall that the letter~$1$ occurs in this kind of block only at
the second position. But here, either $x_k[n]$ is still in
the block, but at least in third position, or it starts the next
block, and hence $x_k[n]=k$. Both cases are incompatible with $x_k[n] = 1$.
Therefore, we have $L_k(F_k(n)) = n$. Hence
\[
L_k(F_k(n+1) - 1) = L_k(F_k(n)) = n < n+1,
\]
and also the left inequality in \eqref{e:Fkj} holds here for $m = n+1$, $j=1$.

\smallskip\noindent
\textbf{Case} $j = 1$, $F_k^k(n) = F_k^k(n-1)+1$.
If $F_k^k(n) - F_k^k(n-1) = 1$, then $F_k(n+1) = F_k(n)$.
In this case, the left inequality in \eqref{e:Fkj} is clear from \eqref{IHn}:
\[
L_k(F_k(n+1) - 1) = L_k(F_k(n) - 1) < n < n+1.
\]
From \eqref{IHnk} and \eqref{IHpred}, we deduce that
$L_k^k(F_k^k(n-1)) = n-1$. Hence, $x_k[n-1]x_k[n]$ starts a new
$\tau_k^k$ block, and thus $x_k[n-1]x_k[n] = k1$. This means that the
left inequality of \eqref{IHn} is actually strict here, because the
letter 1 cannot start the word $\tau_k(x_k[F_k(n)])$. So finally
\[
n < L_k(F_k(n)) = L_k(F_k(n+1)),
\]
and we have finished the proof of \eqref{e:Fkj} for $m = n+1$,
$j = 1$, in all possible cases.

\smallskip\noindent
\textbf{Case $j \ge 2$.}
It remains to prove \eqref{e:Fkj} for $m = n+1$ and $j \ge 2$.
Since $1 \le F_k(n+1) \le n$ by Proposition~\ref{basicF}, the
induction hypothesis
holds for $m = F_k(n+1)$ and $j-1$. Hence
\[
L_k^{j-1}(F_k^j(n+1) - 1) \le F_k(n+1)-1 \le
F_k(n+1) \le L_k^{j-1}(F_k^j(n+1)).
\]
Applying~$L_k$, which preserves these inequalities by monotonicity (Proposition~\ref{basicL}), gives
\[
L_k^j(F_k^j(n+1) - 1) \le L_k(F_k(n+1) - 1) \le
L_k(F_k(n+1)) \le L_k^j(F_k^j(n+1)).
\]
Finally, we can use \eqref{e:Fkj}, which we proved for $m = n+1$ and
$j=1$ above, in the middle:
\[
L_k^j(F_k^j(n+1) - 1) \le L_k(F_k(n+1) - 1) < n+1 \le L_k(F_k(n+1)) \le L_k^j(F_k^j(n+1)).
\]
Therefore, \eqref{e:Fkj} indeed holds for $m = n+1$ and all $j \ge 1$. Thus, by induction, \eqref{e:Fkj} holds for all $m \ge 1$, which concludes the proof of the theorem.
\end{proof}

In particular, we have $F_k^j(L_k^j(n)) = n$ for all $k\ge 1$ and
$j,n \ge 0$.
Moreover, $L_k^j(n)$ is the largest element of $F_k^{-j}(\{n\})$ while $L_k^j(n-1)+1$ is the smallest one (for $n>0$), and these
extrema may coincide. In particular, this is always the case when
$j=0$ and quite frequently when $j=1$ and $k>1$; see
Proposition~\ref{p:one-antecedent} for a study of this ratio.

The relationship between $F_k^j$ and $L_k^j$ can also be formulated as a Galois connection~\cite{Erne1993}.

\begin{corollary}\label{galois}
For all $j\ge 0$ and $k\ge 1$, the functions $F_k^j$ and~$L_k^j$
form a \emph{Galois connection} between~$\mathbb{N}$ and itself
(with $F_k^j$ as \emph{left adjoint} and $L_k^j$ as \emph{right adjoint}).
Indeed, for all $m,n \ge 0$ we have $F_k^j(n) \le m$ if and only if $n \le
L_k^j(m)$. Moreover, this Galois connection is said to be a
\emph{Galois insertion} because $F_k^j\circ L_k^j = \mathrm{id}$.
\end{corollary}

\begin{proof}
The case $n=0$ is obvious. For $n \ge 1$,
if $F_k^j(n) \le m$, then $n \le L_k^j(F_k^j(n)) \le L_k^j(m)$ by
Theorem~\ref{t:1} and the monotonicity of~$L_k^j$. Conversely, if
$n \le L_k^j(m)$, then $F_k^j(n) \le F_k^j(L_k^j(m))$ by monotonicity of~$L_k^j$ and $F_k^j(L_k^j(m)) = m$ as seen above.
\end{proof}

Note that Proposition~\ref{p:LCF} below gives a nice
expression of~$L_k(n)$ as $n+F_k^{k-1}(n)$. Hence
\[
L_k(F_k(n)) = F_k(n) + F_k^k(n) = n+1-\partial F_k(n) \in \{n,n+1\}.
\]

\section{Counting letters}
\label{s:count}

We now express the number of occurrences for the letters $1,\ldots,k$ in
prefixes of~$x_k$. Thanks to Theorem~\ref{t:1}, this relates
them to functions~$F_k$ in various ways.

Let $C_k^{(P)}(n)$ denote the number of letters satisfying the
predicate~$P$ in the prefix $x_k[0{:}n)$. More formally,
\[
C_k^{(P)}(n) = \#\{0 \le j < n \,:\, P(x_k[j])\}.
\]
In particular, we use the expressions
\begin{itemize}
\item $C_k^{(=i)}(n)$ for counting the occurrences of a specific letter $i$,
\item $C_k^{(>i)}(n)$ = $C_k^{(=i+1)}(n)+\cdots+C_k^{(=k)}(n)$
for counting all letters strictly greater than $i$.
\end{itemize}

\begin{proposition}\label{p:counting}
For all $k \ge 1$ and $n\ge 0$, we have
\begin{align}
F_k^{k-1}(n) &= C_k^{(=k)}(n), \label{e:FCk} \\
F_k^j(n) & = C_k^{(>j)}(n) & & \hspace{-8em} \text{for all}\ 0 \le j < k,
\label{e:FC>j} \\
F_k^{k+i-1}(n) & = C_k^{(=i)}(n+i) & & \hspace{-8em} \text{for all}\ 1 \le i < k.
\label{e:FCi}
\end{align}
\end{proposition}

\begin{proof}
We can use the same counting technique to prove each of these three equations.

For Equation~\eqref{e:FCk}, we already mentioned that
$x_k$ can be seen as a succession of ``blocks''
$\tau_k^{k-1}(i) = k1\cdots(i-1)$, each one containing~$k$ only as first
letter.
For a given $m\ge 0$, the letter $x_k[m]$ belongs to one of these blocks,
say the $p$-th one (counting from $p=1$ for the first block).
We hence have $p$
occurrences of the letter $k$ in $x_k[0]\cdots x_k[m]$, so
$C_k^{(=k)}(m+1)=p$.
Also note that the first $p$ blocks have a total length of
$\big|\tau_k^{k-1}(x_k[0{:}p))\big| = L_k^{k-1}(p)$. So $L_k^{k-1}(p)$
is also the first index in the $(p+1)$-st block, and hence
$L_k^{k-1}(p) > m$. Similarly $L_k^{k-1}(p-1)$ is
the first index of the $p$-th block, so all in all
\[
L_k^{k-1}(p-1) \le m < L_k^{k-1}(p).
\]
After substituting~$p$ and setting $n=m+1$, we obtain that for all $n>0$ we have
\[
L_k^{k-1}(C_k^{(=k)}(n) - 1) < n \le L_k^{k-1}(C_k^{(=k)}(n)).
\]
By Theorem~\ref{t:1} and Proposition~\ref{p:bound_unique}, this implies
$F_k^{k-1}(n) = C_k^{(=k)}(n)$
for all $k,n \ge 1$. Moreover, this
identity trivially holds for $n=0$ as well.

For Equation~\eqref{e:FC>j}, we generalize the previous
counting technique.
For $0\le j< k$ and a letter $1\le i\le k$, the word $\tau_k^j(i)$ starts with
exactly one letter strictly greater than~$j$, while the rest of this word consists 
of letters less than or equal to~$j$. Indeed, either $i+j\le k$ and
$\tau_k^j(i)=i+j > j$, or $i+j\ge k$ and
$\tau_k^j(i) = \tau_k^{i+j-k}(k) = k1\cdots (i+j-k)$ with $i+j-k\le j$.
Just as before, we deduce that
\[
L_k^{j}(C_k^{(>j)}(n) - 1) < n \le L_k^{j}(C_k^{(>j)}(n))
\]
for all $0\le j < k$ and $n>0$. As earlier, this allows establishing
$F_k^j(n) = C_k^{(>j)}(n)$
for all $0\le j< k$ and $n>0$. Once again,
this identity also holds for $n\ge 0$.

Finally, we use yet another instance of the same technique
for proving Equation~\eqref{e:FCi}.
Consider $1 \le i < k$. For all letters $1\le j\le k$ the words
$\tau_k^{k+i-1}(j)$
contain the letter~$i$ only at the $(i+1)$-st position. Indeed,
these words can also be written $\tau_k^{j+i-1}(k)$ (since
$\tau_k^{k-j}(j) = k$), so they all admit $\tau_k^i(k)=k1\cdots i$ as
a common prefix, possibly followed first by letters greater than $i$
and then by new blocks no larger than $k1\cdots(i-1)$. For similar
reasons as before, we hence have
\[
L_k^{k+i-1}(C_k^{(=i)}(n) - 1) < n-i \le L_k^{k+i-1}(C_k^{(=i)}(n))
\]
for all $n \ge i+1$.
So for all $1 \le i < k$ and $n\ge 0$, we have
$F_k^{k+i-1}(n) = C_k^{(=i)}(n+i)$.
\end{proof}

Actually, Equation~\eqref{e:FCk} can also be deduced from
Equation~\eqref{e:FC>j} in the particular case $j=k-1$, since
$C_k^{(>k-1)} = C_k^{(=k)}$.
Also note that Equation~\eqref{e:FCi} could be extended to the case
$i=0$ if we replace the letter $k$ by $0$ in the substitution~$\tau_k$ and
its fixed point~$x_k$.

We can now give an interesting alternative expression for~$L_k(n)$ (which is the largest element of $F_k^{-1}(\{n\})$ by Theorem~\ref{t:1}).

\begin{proposition}\label{p:LCF}
For all $k\ge 1$ and $n\ge 0$,
\[
L_k(n) = n + C_k^{(=k)}(n) = n + F_k^{k-1}(n).
\]
\end{proposition}
\begin{proof}
We have $|\tau_k(k)|=2$ while $|\tau_k(i)|=1$ for the other letters
$i\neq k$. Hence
\[
L_k(n) = 2\,C_k^{(=k)}(n) + C_k^{(\neq k)}(n) = n + C_k^{(=k)}.
\]
Finally, Equation~\eqref{e:FCk} leads to the last equality.
\end{proof}

By Equation~\eqref{e:FC>j}, we can also express the number of occurrences of a
specific letter $1\le j<k$ via the
difference between $C_k^{(>j-1)}$ and $C_k^{(>j)}$.
Hence for $1\le j<k$ and $n\ge 0$
\begin{equation}
\label{e:FCj}
F_k^{j-1}(n)-F_k^j(n) = C_k^{(>j-1)}(n)-C_k^{(>j)}(n) = C_k^{(=j)}(n).
\end{equation}
In particular, we obtain the following proposition.

\begin{proposition}
\label{FC1}
For $k>1$ and $n\ge 0$,
\[
F_k(n) = n-C_k^{(=1)}(n) = C_k^{(\neq 1)}(n).
\]
Hence $\partial F_k(n)$ is 0 if and only if $x_k[n]=1$, and $1$ otherwise.
\end{proposition}
\begin{proof}
Direct use of the previous equation in the particular case $j=1$.
Alternatively, one may use Equation~\eqref{e:FCi} for $i=1$, and then
the recursive definition of $F_k$.

Afterwards, we get
$F_k(n+1) - F_k(n) = 1 - (C_k^{(=1)}(n+1) - C_k^{(=1)}(n))$,
which is $0$ if and only if $x_k[n]=1$ and $1$ otherwise.
\end{proof}

This important link between $\partial F_k$ and $x_k$ helps us to transfer
many properties from one to the other.
In particular, $\partial F_k$ cannot have
two consecutive zeros, and it admits up to~$k$ consecutive ones but not~$k+1$.

Now let us describe the letters of~$x_k$ in terms of
differences~$\partial F_k^j(n)$.
Recall from Proposition~\ref{basicF} that
these differences are always either 0 or~1. Of course,
$\partial F_k^0(n) = 1$ and $\partial F_k^j(0) = 1$.
Now, for a non-zero~$n$, we have $F_k^j(n) = 1$ when $j$ is large enough,
in particular for $j\ge n-1$.
Indeed, on a non-zero argument, $F_k$ either returns~1 or removes at
least one from its argument, and this is iterated here $j$ times.
As a consequence, we always have $\partial F_k^n(n) = 0$ when $n>0$.
Moreover, if for some~$j$, we have $\partial F_k^j(n) = 0$, then
$\partial F_k^{j+1}(n) = 0$ as well. Conversely, if
$\partial F_k^j(n) = 1$, then $\partial F_k^{j-1}(n) = 1$.
Hence for all $k,n \ge 1$, the sequence $(\partial F_k^j(n))_{j\in\mathbb{N}}$
consists of a block of ones followed by an infinity of zeros.
Actually, the letters of~$x_k$ indicate how deep to dive in these
differences to find a first zero (or give up after $k-1$ attempts).

\begin{proposition}
\label{diffFx}
Let $1\le j < k$, $n\ge 0$. We have $x_k[n]=j$ if and only if
both $\partial F_k^{j-1}(n) = 1$ and $\partial F_k^j (n) = 0$.
Moreover, for $k\ge 1$, we have $x_k[n]=k$ if and only if $\partial F_k^{k-1}(n) = 1$.
(In the latter case, $\partial F_k^{k}(n)$ could be either 0 or~1.)

\end{proposition}
\begin{proof}
First, this statement is obvious for $k=1$. We now assume that $k>1$.
By subtracting Equation~\eqref{e:FCj} for $n+1$ and~$n$, we obtain
for each $1 \le j < k$ that
\[ \partial F_k^{j-1}(n) - \partial F_k^j(n) =
  C_k^{(=j)}(n+1)-C_k^{(=j)}(n).
\]
This amounts to~1 if and only if $x_k[n]=j$, and 0 otherwise.
So in particular for each $i<x_k[n]$ we have
$\partial F_k^{i-1}(n) = \partial F_k^{i}(n)$.
Meanwhile, we just noticed in Proposition~\ref{FC1} that
$\partial F_k(n) = 0$ if and only if $x_k[n]=1$. Write $\ell = x_k[n]$.
Three situations may occur:
\begin{itemize}
\item Either $\ell=1$, and we directly have $\partial F_k(n) = 0$ and
$\partial F_k^0(n) = 1$.
\item Or $1<\ell<k$ and $\partial F_k(n) = 1$. We can propagate
$\partial F_k^i(n) = 1$ for all $i<\ell$ and finish with
$\partial F_k^{\ell-1}(n) - \partial F_k^{\ell}(n) = 1$.
Hence $\partial F_k^\ell(n) = 0$.
\item Lastly, if $\ell=k$, then the propagation $\partial F_k^i(n) = 1$ goes up to $i \le k-1$.
Hence the desired statement follows. \qedhere
\end{itemize}
\end{proof}

\section{Related polynomials and algebraic integers}
\label{s:poly}

We now introduce two families of polynomials whose positive zeros
appear as average slopes for $F_k$ and~$L_k$
and letter frequencies for~$x_k$ (Section~\ref{s:infini}).

\begin{definition} \label{d:alphabeta}
For $k\ge 1$, we define $P_k(X) = X^k+X-1$ and
$Q_k(X) = X^k-X^{k-1}-1$. We let $\alpha_k$ (and $\beta_k$ respectively) be the
unique positive zero of $P_k$ (and $Q_k$ respectively). Note that $\beta_k = 1/\alpha_k$.
\end{definition}

A more complete investigation of these polynomials can be found in
Selmer~\cite{Selmer56} (where the polynomials $X^k\pm(X+1)$
lead to $P_k(-X)$) and Dilcher~\cite{Dilcher1993}.
We now state the elementary results needed in the rest of our work.
In particular, for $k\ge 1$, the polynomial $P_k$ is strictly increasing on
$\mathbb{R_+}$ and admits exactly one zero there. This zero $\alpha_k$
clearly satisfies
$\frac{1}{2}\le \alpha_k < 1$. Moreover the polynomial $P_k$ is the
negation of the reciprocal polynomial of~$Q_k$, i.e.,
$P_k(X)=-X^k\,Q_k(1/X)$ (and vice versa).
As a consequence, the zeros of~$P_k$ are the inverse of the zeros of~$Q_k$
and vice versa. In particular $\beta_k = 1/\alpha_k$ is indeed the
unique zero of $P_k$ in $\mathbb{R_+}$ and furthermore $1 < \beta_k \le 2$.

Figure~\ref{f:alphabeta} gives approximate values for the first
$\alpha_k$ and~$\beta_k$. In particular, note that $\beta_1 = 2$ and
$\beta_2$ is the golden ratio $\varphi = \frac{1+\sqrt 5}{2}$.
Thanks to the rational zero theorem, one can easily show that
$\alpha_k$ and~$\beta_k$ are irrational for $k\ge 2$.

\begin{figure}[ht]
\[
\begin{aligned}[c]
 \alpha_1 &= 0.5 \\
 \alpha_2 &= 0.6180339887498948\ldots \\
 \alpha_3 &= 0.6823278038280193\ldots \\
 \alpha_4 &= 0.7244919590005157\ldots \\
 \alpha_5 &= 0.7548776662466925\ldots \\
 \alpha_6 &= 0.7780895986786012\ldots
\end{aligned}
\qquad
\begin{aligned}[c]
 \beta_1 &= 2 \\
 \beta_2 &= 1.618033988749895\ldots \\
 \beta_3 &= 1.465571231876768\ldots \\
 \beta_4 &= 1.380277569097614\ldots \\
 \beta_5 &= 1.324717957244746\ldots \\
 \beta_6 &= 1.285199033245349\ldots
\end{aligned}
\]
\caption{The first $\alpha_k$ and $\beta_k$,
    positive zeros of $X^k + X - 1$ and $X^k - X^{k-1} - 1$.}
\label{f:alphabeta}
\end{figure}

\begin{proposition}
\label{p:alpha-incr}
The sequence $(\alpha_k)_{k\in\mathbb{N}_+}$ is strictly increasing while 
$(\beta_k)_{k\in\mathbb{N}_+}$ is strictly decreasing.
Moreover, for $k\ge 1$, we have $1 + \frac{1}{k} \le \beta_k \le 1 + \frac{1}{\sqrt{k}}$
or equivalently $\sqrt{k} \le \beta_k^{k-1} \le k$.
Hence the $(\beta_k)$ sequence converges to $1$, as well as the
$(\alpha_k)$ sequence.
Actually, the following asymptotic expressions hold for $k\to\infty$:
\begin{align}
\alpha_k &= 1 - \frac{\ln{k}}{k} + o\bigg(\frac{\ln{k}}{k}\bigg),
\label{e:oalpha} \\
\beta_k &= 1 + \frac{\ln{k}}{k} + o\bigg(\frac{\ln{k}}{k}\bigg).
\label{e:obeta}
\end{align}
\end{proposition}

\begin{proof}
Let $k\ge 1$. Suppose $\alpha_{k+1} \le \alpha_{k}$. Since
$\alpha_k$ and $\alpha_{k+1}$ are in $(0,1)$, we would have $\alpha_{k+1}^{k+1} = \alpha_{k+1}^k \alpha_{k+1} < \alpha_{k+1}^k \le \alpha_k^k$, leading to $0 = P_{k+1}(\alpha_{k+1}) < P_k(\alpha_k) = 0$, a
contradiction.
Hence $(\alpha_k)$ is strictly increasing and $(\beta_k)=(\alpha_k^{-1})$
is strictly decreasing.

Now $Q_k(\beta_k) = 0$ can be reformulated as
$\beta_k^{k-1}(\beta_k - 1) = 1$ and hence
\begin{equation} \label{e:alphakm1}
\alpha_k^{k-1} = \beta_k^{-(k-1)} = \beta_k-1.
\end{equation}
This provides the
equivalence between $\sqrt{k} \le \beta_k^{k-1} \le k$ and
$1 + \frac{1}{k} \le \beta_k \le 1 + \frac{1}{\sqrt{k}}$.

For proving the lower bound $\sqrt{k} \le \beta_k^{k-1}$,
it is sufficient to note that
\begin{equation} \label{e:betageom}
1 + \beta_k + \cdots + \beta_k^{k-1} =
\frac{\beta_k^k-1}{\beta_k-1} = \frac{\beta_k^{k-1}}{\beta_k^{-(k-1)}}
  = \beta_k^{2(k-1)}.
\end{equation}
Each term on the left of Equation~\eqref{e:betageom} is 1 or larger,
so $k \le \beta_k^{2(k-1)}$, hence the desired lower bound.
For the upper bound $\beta_k^{k-1} \le k$, we divide
by $\beta_k^{k-1}$ in Equation~\eqref{e:betageom} and obtain
\[
\beta_k^{k-1} = \frac{\beta_k^{2(k-1)}}{\beta_k^{k-1}}
= \frac{1}{\beta_k^{k-1}} + \cdots + \frac{\beta_k^{k-1}}{\beta_k^{k-1}}
= \alpha_k^{k-1} + \cdots + \alpha_k + 1 \le k.
\]

The two asymptotic expressions \eqref{e:oalpha} and \eqref{e:obeta}
are equivalent: considering the multiplicative inverse of any of these
expressions leads to the other one.
An approximation related to \eqref{e:oalpha} appears in
Selmer~\cite[Eq.~4.4]{Selmer56} but without further explanation,
while \eqref{e:obeta} is a direct consequence of
Dilcher~\cite[Lem.~3]{Dilcher1993}. Alternatively, we give here a
simple direct proof of \eqref{e:oalpha}.
For $k\ge 1$, we consider the function $h_k(x)=\sqrt[k]{1-x}$ and note that
$\alpha_k$ is a fixed point of $h_k$. It could be proved that
iterating $h_k$ on any initial point in $(0,1)$ converges to
$\alpha_k$, but here it suffices to use a few initial
terms of such a sequence, namely $u_0=1-1/e$ and $u_1 = h_k(u_0)$ and
$u_2 = h_k(u_1)$ and $u_3 = h_k(u_2)$.
For these values, it is straightforward to compute
\begin{align*}
u_1 &= e^{-1/k} = 1 - \frac{1}{k} + o\bigg(\frac{1}{k}\bigg), \\
u_2 &= 1 - \frac{\ln{k}}{k} + o\bigg(\frac{1}{k}\bigg), \\
u_3 &= 1 - \frac{\ln{k}}{k} + \frac{\ln{\ln{k}}}{k} +
o\bigg(\frac{1}{k}\bigg) .
\end{align*}
Meanwhile, a basic study of the function
$h_k$ ensures that $u_0 < u_2 < \alpha_k < u_3 < u_1$ when $k \ge 3$. 
The desired expression \eqref{e:oalpha} follows by weakening
the previous asymptotic expressions for both $u_2$ and $u_3$
to $1-\ln{k}/k+o(\ln{k}/k)$.

\end{proof}

\section{Infinitary behavior}
\label{s:infini}

For each $k\ge 1$, the substitution~$\tau_k$ is a \emph{primitive}
morphism~\cite{AlloucheCSTaxonomy};
i.e., there exists an exponent $p\ge 1$ such that
for all letters $1\le i,j \le k$, the letter
$i$ occurs in $\tau_k^{p}(j)$. Here, the first adequate exponent is $p=2k-2$.
Indeed, $\tau_k^{2k-2}(j)$ admits $\tau_k^{k-1}(k) = k1\cdots (k-1)$ as
a prefix because $\tau_k^{k-1}(j)$ has $k$ as first letter.
Since $x_k$ is the fixed point of a primitive morphism, each 
letter~$i$, $1\le i \le k$, has a \emph{frequency}, which we denote by $\freq_k(i)$.
More precisely, the limit $\freq_k(i) = \lim_{n\to\infty} \tfrac{1}{n} C_k^{(=i)}(n)$
exists, we have $0 < \freq_k(i) \le 1$, and there exists a formula
for computing it \cite[Thm.~8.4.7]{allouche_shallit_2003}. 
Here, an easier approach for computing
$\freq_k(i)$ is to consider~$F_k$.

\begin{theorem}
\label{t:limits}
For $j\ge 0$ and $k\ge 1$, the following limits exist and have the
given values, where $\alpha_k$ and $\beta_k$ come from
Definition~\ref{d:alphabeta}:
\begin{align*}
\lim_{n\to\infty} \tfrac{1}{n}F_k^j(n) &= \alpha_k^j, \\
\lim_{n\to\infty} \tfrac{1}{n}L_k^j(n) &= \beta_k^j, \\
\freq_k(i) &= \alpha_k^{k+i-1}, \hspace{5em} {\text for}\ 1 \le i < k, \\
\freq_k(k) &= \alpha_k^{k-1} = \beta_k - 1.
\end{align*}
\end{theorem}

\begin{proof}
Recall from Proposition~\ref{FC1} that for $k>1$ and $n\ge 0$,
we have $F_k(n) = n-C_k^{(=1)}(n)$. Hence the limit
$\lim_{n\to\infty} \tfrac{1}{n} F_k(n)$ exists as
well\footnote{Surprisingly, we have not found any obvious methods for
proving this convergence of $\tfrac{1}{n} F_k(n)$ directly from the
recursive definition of $F_k$, without using this ``detour'' via
words.}
and is $1 - \freq_k(1)$.
In the case $k=1$, we have $F_1(n)=\lceil n/2\rceil$. Hence
$\lim_{n\to\infty} \tfrac{1}{n} F_1(n) = \tfrac{1}{2}$.

Now that $\frac{1}{n} F_k(n)$ is known to converge to some finite
limit~$\ell$, computing this limit is quite
straightforward. Indeed,
the recursive equation of~$F_k$ can be
reformulated as
\[
\frac{F_k(n)}{n} = 1 -
\frac{F_k(F_k^{k-1}(n-1))}{F_k^{k-1}(n-1)}\cdots\frac{F_k(n-1)}{n-1}\,\frac{n-1}{n}
\]
for $n>1$. Each fraction can be shown to converge to~$\ell$,
except $\tfrac{n-1}{n}$, which tends to~1.
Hence $\ell = 1  -  \ell^k$ and obviously $\ell$ is a nonnegative real
number. Hence $\ell = \alpha_k$.
As a consequence, for $k\ge 2$, the frequency $\freq_k(1)$ is
$1 - \alpha_k = \alpha_k^k$ (and~$1$ when $k=1$).
The same telescoping technique gives
$\lim_{n\to\infty}\tfrac{1}{n}F_k^j(n) = \alpha_k^j$.

Concerning $L_k^j$, a consequence of Theorem~\ref{t:1} is
$F_k^j(L_k^j(n)) = n$. Hence
\[
\frac{L_k^j(n)}{n} = \bigg(\frac{F_k^j(L_k^j(n))}{L_k^j(n)}\bigg)^{-1}.
\]
Since $L_k^{j}(n)\ge n$, this implies that $\tfrac{1}{n} L_k^j(n)$
converges, and its limit is $\alpha_k^{-j} = \beta_k^j$.

For the frequency of the letter $k$, Equation~\eqref{e:FCk} implies
that $\freq_k(k) = \alpha_k^{k-1}$, which is also $\beta_k - 1$ by
Equation~\eqref{e:alphakm1}.
And for the frequencies of
the other letters $1 \le i < k$, one may exploit either
Equation~\eqref{e:FCj} or Equation~\eqref{e:FCi}. For instance, the
former leads to
\[
\freq_k(i) = \alpha_k^{i-1}-\alpha_k^i = \alpha_k^{i-1} (1 - \alpha_k) =
\alpha_k^{k+i-1}.
\]
In particular, this subsumes the case $i=1$ seen earlier.
Finally, one may check that the sum of all these frequencies, from
$\alpha_k^{k-1}$ (letter~$k$) to $\alpha_k^{2k-2}$ (letter~$k-1$), is
of course~1.
\end{proof}

\begin{definition}
For two functions $f,g: \mathbb{N}\to\mathbb{N}$, we say that
$f$ is \emph{ultimately smaller} than~$g$ and write $f <_\infty g$
when there exists $N$ such that $f(n) < g(n)$ for all $n\ge N$.
\end{definition}

\begin{corollary}\label{F_lt_infty}
When $\alpha_k^j < \alpha_{k'}^{j'}$ for some
$k,k'\ge 1$ and $j,j'\ge 0$, then $F_k^j <_\infty F_{k'}^{j'}$.
\end{corollary}

\begin{proof}
Theorem~\ref{t:limits} gives $\lim_{n\to\infty} \frac{1}{n}F_k^j(n) =\alpha_k^j < \alpha_{k'}^{j'} =
 \lim_{n\to\infty}\frac{1}{n}F_{k'}^{j'}(n)$.
For $n$ large enough, both sides become close enough to their
limits, so there must exists $N\ge 1$ such that $n \ge N$ implies
$\tfrac{1}{n}F_k^j(n) < \tfrac{1}{n}F_{k'}^{j'}(n)$ and hence
$F_k^j(n) < F_{k'}^{j'}(n)$.
\end{proof}

\begin{corollary}\label{Fk_lt_FSk_eventually}
For $k\ge 1$, we have $F_k <_\infty F_{k+1}$,
$F_k^k >_\infty F_{k+1}^{k+1}$ and
$F_k^{k+1} >_\infty F_{k+1}^{k+2}$.
More generally, for $j\ge 0$ and $k \ge 1$:
\begin{enumerate}[label=(\alph*)]
\item we have $F_k^{k+j} >_\infty F_{k+1}^{k+j+1}$ when
$\alpha_{k} \ge \tfrac{j}{j+1}$ (hence in particular when $k \ge j^2$);
\item we have $F_k^{k+j} <_\infty F_{k+1}^{k+j+1}$ when
$\alpha_{k+1} \le \tfrac{j}{j+1}$ (hence in particular when $1\le k < j$).
\end{enumerate}
\end{corollary}

\begin{proof}
All these facts are obtained by the previous corollary; we just have
to compare the corresponding average slopes.
For $F_k <_\infty F_{k+1}$, Proposition~\ref{p:alpha-incr}
directly gives $\alpha_k < \alpha_{k+1}$. For $F_k^k >_\infty F_{k+1}^{k+1}$,
we have $\alpha_k^k = 1  -  \alpha_k > 1  -  \alpha_{k+1} = \alpha_{k+1}^{k+1}$.
Now, for $F_k^{k+1} >_\infty F_{k+1}^{k+2}$, we have
\[
\alpha_k^{k+1} - \alpha_{k+1}^{k+2}
= (1 - \alpha_k)\alpha_k - (1 - \alpha_{k+1})\alpha_{k+1}
= (\alpha_{k+1} - \alpha_k)(\alpha_k + \alpha_{k+1} - 1) > 0
\]
since $\tfrac{1}{2} \le \alpha_{k} < \alpha_{k+1}$.
More generally, let $j\ge 0$. As before,
\[
\alpha_k^{k+j} - \alpha_{k+1}^{k+j+1}
= (1 - \alpha_k)\alpha_k^j-(1 - \alpha_{k+1})\alpha_{k+1}^{j}.
\]
The function $(1 - X)X^j$ is strictly increasing between
$0$ and~$\frac{j}{j+1}$ and strictly decreasing afterwards.
Since we always have $\alpha_{k} < \alpha_{k+1}$, then
$\alpha_k^{k+j} - \alpha_{k+1}^{k+j+1} > 0$ at least when
$\alpha_{k} \ge \frac{j}{j+1}$ or equivalently
when $\beta_{k} \le 1+\frac{1}{j}$.
Thanks to the bounds in Proposition~\ref{p:alpha-incr},
this happens at least when $\sqrt{k} \ge j$, i.e., $k \ge j^2$.
Conversely,
$\alpha_k^{k+j} - \alpha_{k+1}^{k+j+1} < 0$ at least when
$\alpha_{k+1} \le \tfrac{j}{j+1}$, for which a sufficient condition
is $1 \le k < j$, still thanks to the bounds in
Proposition~\ref{p:alpha-incr}.
\end{proof}

Corollary~\ref{Fk_lt_FSk_eventually} is to be compared with
the results of the next section, for instance Theorem~\ref{t:2a}:
The latter gives only non-strict inequalities, but for all $n\ge 0$,
while here we have proved strict inequalities, but only when $n$ is
larger than some bounds. Moreover, we do not have explicit estimates for
these bounds for the moment.

As a related matter, we can estimate the ratio of numbers having
unique preimages by~$F_k$.

\begin{proposition}
\label{p:one-antecedent}
For $k\ge 1$, $n\ge 0$, let us call
$U_k(n) = \#\{ 0 \le j < n: |F_k^{-1}(\{j\})|=1 \}$.
Then, for $n>0$, we have
\[
U_k(n) = 2n-1-L_k(n-1) = n - F_k^{k-1}(n-1),
\]
and hence
\[
\lim_{n\to\infty} \tfrac{1}{n}U_k(n) = 1 - \alpha_k^{k-1} = 2 - \beta_k.
\]
\end{proposition}

\begin{proof}
Let $n>0$.
We already know from Theorem~\ref{t:1} and Proposition~\ref{p:LCF}
that the largest element of $F_k^{-1}(\{n-1\})$ is $L_k(n-1) = n-1+F_k^{k-1}(n-1)$.
There are $U_k(n)$ numbers between $0$ and~$n-1$ with a unique preimage, and the other $n - U_k(n)$ numbers have exactly two preimages; see the discussion after
Proposition~\ref{FC1}. By counting the cardinality of 
$F_k^{-1}(\{0,\ldots,n-1\})$, we obtain
\[
U_k(n) + 2(n - U_k(n)) = 1 + L_k(n-1) = n + F_k^{k-1}(n-1).
\]
Combined with Theorem~\ref{t:limits}, this leads to the desired
equations and limits.
\end{proof}

In particular, for $k=1$ and $n>0$ we get $U_1(n) = 1$.
Indeed, since $F_1(n) = \lceil \tfrac{n}{2} \rceil$, the only number with a unique preimage by $F_1$ is~$0$.
Then $\tfrac{1}{n}U_k(n)$ tends to $0.3819\ldots$ for $k=2$,
to $0.5344\ldots$ for $k=3$, and to $0.6197\ldots$ for $k=4$, and these limits
tend to~$1$ as $k$ grows.

\section{Monotonicity over the parameter \texorpdfstring{$k$}{\it k}}
\label{s:kmono}

This section studies the monotonicity of $L_k$ and~$F_k$ when the
parameter $k$ varies. In all this section, we compare functions via
\emph{pointwise order}: an inequality such as $F_k \le F_{k+1}$ means
that $F_k(n) \le F_{k+1}(n)$ \emph{for all} points $n\ge 0$.
As such, the results here are quite stronger than the ones of the
previous section about the infinitary behavior, i.e., when $n$ is large enough.

First, we state a nice duality between functions $F_k^j$ and
$L_k^j$ with respect to this pointwise order, extending the results
of Section~\ref{s:relating_F_L}. Thanks to this, all the following technical lemmas about $L_k$ have immediate counterparts for~$F_k$.

\begin{proposition}\label{galoisbis}
For $j,j'\ge 0$ and $k,k'\ge 1$, we have
$L_k^j \le L_{k'}^{j'}$ if and only if \mbox{$F_{k'}^{j'} \le F_k^j$}.
Furthermore, we can be more precise concerning the relative positions
where these inequalities occur. For $n\ge 0$, we have
\begin{enumerate}[label=(\alph*)]
\itemsep.5ex
\item $L_k^j(n) \le L_{k'}^{j'}(n)$ if and only if
 $F_{k'}^{j'}(m) \le F_k^j(m)$, with $m = L_k^j(n)$,
\item $L_k^j(n) < L_{k'}^{j'}(n)$ if and only if
 $F_{k'}^{j'}(m) < F_k^j(m)$, with $m = L_{k'}^{j'}(n)$,
\item $L_k^j(m) \le L_{k'}^{j'}(m)$ for $m = F_k^j(n)$ implies that $F_{k'}^{j'}(n) \le F_k^j(n)$,
\item $F_k^j(n) < F_{k'}^{j'}(n)$ implies that 
$L_{k'}^{j'}(m) < L_k^j(m)$ for $m = F_{k'}^{j'}(n)$.
\end{enumerate}
\end{proposition}

\begin{proof}
First, the fact that $L_k^j \le L_{k'}^{j'}$ implies $F_{k'}^{j'} \le F_k^j$
is a consequence of point~(c), while point~(a) implies the
other direction. Also note that (b) is a contrapositive
version of~(a), and the same for (d) and~(c).

For proving point~(a), let $n\ge 0$.
Then $L_k^j(n) \le L_{k'}^{j'}(n)$ is equivalent to
$F_{k'}^{j'}(L_{k}^{j}(n)) \le n$ by Corollary~\ref{galois}, and we have $n = F_k^j(L_k^j(n))$ by Theorem~\ref{t:1}.

For point~(c), we use that $F_{k'}^{j'}(n) \le F_{k}^{j}(n)$
is equivalent to $n \le L_{k'}^{j'}(F_{k}^{j}(n))$ by Corollary~\ref{galois}.
By Theorem~\ref{t:1} (or trivially when $n=0$), we have $n \le L_{k}^{j}(F_{k}^{j}(n))$, and thus $n \le L_{k'}^{j'}(F_{k}^{j}(n))$ when $L_k^j(F_k^j(n)) \le L_{k'}^{j'}(F_k^j(n))$.
\end{proof}

Let us now study the ordering of the functions $L_k^j$.
Note first that
\[
|\tau_k^k(i)| = i+1 = |\tau_k^{k-1}(i)| + 1 \quad \mbox{for all}\ 1 \le i \le k.
\]
For all $k \ge 1$, $n \ge 0$, this implies that $L_k^k(n) = L_k^{k-1}(n) + n$ and thus
\begin{equation} \label{e:kk1}
L_{k+1}^{k+1}(n) - L_k^k(n) = L_{k+1}^k(n) - L_k^{k-1}(n).
\end{equation}

One of our main results is that the sequence of functions $(L_k)$ is monotonic, which we prove here by mutual induction with another property comparing iterations $L_k^j$ and~$L_{k+1}^{j+1}$.

\begin{theorem} \label{p:1}
For all $k,n \ge 1$, $0 \le j \le k$, we have
\begin{align}
  L_k(n) &\ge L_{k+1}(n), \quad \mbox{i.e.,}\ C_k^{(=k)}(n) \ge C_{k+1}^{(=k+1)}(n), \label{e:tau} \\[.5ex]
L_k^j(n) &< L_{k+1}^{j+1}(n). \label{e:tauj}
\end{align}
\end{theorem}

\begin{proof}
Let $k \ge 1$. 
Since $x_k[0] = k$, $x_{k+1}[0] = k+1$, and $|\tau_k^j(k)| = j+1 < j+2 = |\tau_{k+1}^{j+1}(k+1)|$ for $0 \le j \le k$,  the inequalities \eqref{e:tau} and~\eqref{e:tauj} hold for $n = 1$, $0 \le j \le k$.

Let $m \ge 2$ and assume that \eqref{e:tau} and \eqref{e:tauj} hold
for all $1 \le n < m$, $0 \le j \le k$.

We first prove \eqref{e:tau} at~$m$, i.e.,
$L_k(m) \ge L_{k+1}(m)$, or equivalently
$F_k^{k-1}(m) \ge F_{k+1}^k(m)$; see Proposition~\ref{p:LCF}.
Let us abbreviate $F_k^{k-1}(m)$ as~$c$ and $F_{k+1}^k(m)$ as~$c'$
and prove $c' \le c$ i.e., $c' - 1 < c$. By Proposition~\ref{basicL},
$L_{k+1}^k$ is strictly increasing, and hence it is sufficient to prove
$L_{k+1}^k(c' - 1) < L_{k+1}^k(c)$. 
Indeed, we have
\begin{equation} \label{e:Lkc}
L_{k+1}^k(c' - 1) < m \le L_{k}^{k-1}(c) \le L_{k+1}^{k}(c),
\end{equation}
where the left and middle inequalities come from Theorem~\ref{t:1}.
To obtain the right inequality in \eqref{e:Lkc}, we distinguish between the cases $k=1$ and $k \ge 2$.
When $k = 1$, then $c = m$ and hence
$L_k^{k-1}(c) = m < L_2(m) = L_{k+1}^k(c)$ by
Proposition~\ref{basicL}~(f). When $k\ge 2$, we can
use~\eqref{e:tauj} for $n = c$, $j = k-1$ because in this case
$1 \le c < m$ by Proposition~\ref{basicF}.

Now let $0 \le h \le k$ and let us prove \eqref{e:tauj} for $n = m$ and
$j = h$. Thanks to Equation~\eqref{e:kk1}, the case $j=k$ is implied
by the case $j=k-1$, so now we can freely assume $h<k$.
If $x_{k+1}[m-1] = k+1$, then \eqref{e:tauj} holds for $n = m$ because it holds for $n = m-1$ and $|\tau_k^h(i)| \le |\tau_k^h(k)| < |\tau_{k+1}^{h+1}(k+1)|$ for all $1 \le i \le k$. 
If $x_{k+1}[m-1] \ne k+1$, then $x_{k+1}[0{:}m) =
\tau_{k+1}(x_{k+1}[0{:}\ell))$ for some $\ell \ge 1$.
Proposition~\ref{basicL} indicates that $\ell < L_{k+1}(\ell) = m$.
Hence \eqref{e:tau} holds for $n=\ell$, and thus
\[
m = L_{k+1}(\ell) \le L_k(\ell).
\]
We apply $L_k^h$ (which is monotonic by Proposition~\ref{basicL}) on
this inequality, and then \eqref{e:tauj} for $n=\ell$ and $j=h+1$, to obtain
\[
L_k^h(m) \le L_k^h(L_k(\ell)) = L_k^{h+1}(\ell) < L_{k+1}^{h+2}(\ell) =
L_{k+1}^{h+1}(L_{k+1}(\ell)) = L_{k+1}^{h+1}(m).
\]
Therefore, \eqref{e:tauj} holds for $n = m$, $0 \le j \le k$.
By induction, \eqref{e:tau} and \eqref{e:tauj} hold for all $n \ge 1$,
$0 \le j \le k$.
\end{proof}

\begin{corollary} \label{c:1bis}
For all $k \ge 1$ and $j \ge 0$, we have $L_k^j \ge L_{k+1}^j$.
\end{corollary}

\begin{proof}
We proceed by induction on~$j$. The case $j=0$ is obvious. The case
$j=1$ is given by Equation~\eqref{e:tau}, trivially extended to $n=0$.
Now assume $L_{k}^j \ge L_{k+1}^j$ for some $j\ge 0$.
For $n \ge 0$, we hence have
\[
L_k^{j+1}(n) = L_k^j(L_k(n)) \ge L_{k+1}^j(L_k(n)) \ge
L_{k+1}^j(L_{k+1}(n)) = L_{k+1}^{j+2}(n),
\]
thanks to the induction hypothesis for~$j$ and
then the monotonicity of~$L_{k+1}^j$, combined with the statement for $j=1$.
We can hence conclude by induction.
\end{proof}

At last, we obtain the following monotonicity of the function sequence $(F_k)$ over the
parameter~$k$, where the special case $j = 1$ is our main result Theorem~\ref{t:2}.

\begin{theorem} \label{t:2a}
For all $j\ge 0$, $k \ge 1$,  we have $F_k^j \le F_{k+1}^j$.
\end{theorem}

\begin{proof}
Consequence of Corollary~\ref{c:1bis} and Proposition~\ref{galoisbis}.
\end{proof}

Comparing $F_k^j$ with $F_{k+1}^{j+1}$ instead of $F_{k+1}^j$, we obtain the opposite inequality.

\begin{theorem} \label{t:2j}
For all $k \ge 1$ and $0 \le j \le k$, we have $F_k^j \ge F_{k+1}^{j+1}$.
\end{theorem}

\begin{proof}
Let $k\ge 1$ and $0 \le j \le k$.
By Proposition~\ref{galoisbis}, the statement to prove is equivalent to
$L_{k+1}^{j+1} \le L_k^j$, which is a direct consequence of
Equation~\eqref{e:tauj} and of $L_{k+1}^{j+1}(0)=0=L_k^j(0)$.
\end{proof}

To sum up, the functions $(F_k^j,\le)$ with their pointwise ordering form
(at least) a nice lattice generated by the following basic cells
for $1\le j\le k$:
\begin{center}
\begin{tikzcd}
  F_{k-1}^{j-1} \arrow[r, "\le"]
    & F_{k}^{j-1} & \\
  & F_{k}^{j}
    \arrow[u, "\ge" {rotate=-90,xshift=0.6em,yshift=0.5em}]
    \arrow[lu, "\ \ge" rotate=-30]
    \arrow[r, "\le"']
    & F_{k+1}^{j} \arrow[lu, "\ \ge"' rotate=-30]
\end{tikzcd}
\end{center}
In such a cell, the vertical edge $F_k^j\le F_k^{j-1}$ is obvious for
sub-linear functions such as~$F_k$. Moreover it is also a double
consequence of $F_k^j \le F_{k-1}^{j-1} \le
F_k^{j-1}$ and
$F_k^j \le F_{k+1}^j \le F_k^{j-1}$.
Also note that the remaining unrelated functions
$F_{k-1}^{j-1}$ and~$F_{k+1}^j$ may actually be incomparable. For instance
$F_3^3(5) = 2 > 1 = F_5^4(5)$, while $F_3^3(9) = 3 < 4 = F_5^4(5)$.
Even if we only retain the infinitary behavior
(as in Section~\ref{s:infini}), the ordering of these functions may vary.
The currently known situation is presented in Figure~\ref{f:Fkj-lattice},
where the edges such as $G\longrightarrow H$ mean $G \le H$.
The row $j=0$ of identity functions
has been omitted here, $\mathrm{id}$~being trivially above all other~$F_k^j$.
The numbers displayed in blue alongside the nodes~$F_k^j$ are approximations
of their average slopes $\lim \frac{1}{n}F_k^j(n) = \alpha_k^j$.
If one of the functions is pointwise below another, their average slopes
are ordered accordingly.\footnote{Conversely, strictly ordered
average slopes only give clues about infinitary behavior, the
functions may well be pointwise incomparable due to early values.}
Also note that the diagonal $j=k$ and first row $j=1$ (in red in the
figure) form an interesting chain of inequalities, with slopes ranging
symmetrically between 0 and~1 (since $1 - \alpha_k^k = \alpha_k$).

\newcommand{\divi}[1]{\lceil\tfrac{n}{#1}\rceil}
\begin{figure}[ht]
\begin{tikzcd}
    & k=1 & k=2 & k=3 & k=4 & k=5 \\
j=1 & \red{F_1=\divi{2}}\rar[red]\dar[leftarrow, "0.5~"' {blue,at start}]
      & \red{F_2=G}\rar[red]\dar[leftarrow, "0.618~"' {blue,at start}]
      & \red{F_3=H}\rar[red]\dar[leftarrow, "0.682~"' {blue,at start}]
      & \red{F_4}\rar[red]\dar[leftarrow, "0.724~"' {blue,at start}]
      & \red{F_5}\dar[leftarrow, "0.754~"' {blue,at start}] \\
j=2 & F_1^2=\divi{4}\rar\dar[leftarrow, "0.25~"' {blue,at start}]
      & \red{F_2^2=G^2}\rar\dar[leftarrow, "0.381~"' {blue,at start}]
        \ular[red]
      & F_3^2=H^2\rar\dar[leftarrow, "0.465~"' {blue,at start}]
        \ular
      & F_4^2\rar\dar[leftarrow, "0.524~"' {blue,at start}]
        \ular
      & F_5^2 \dar[leftarrow, "0.569~"' {blue,at start}]
        \ular\\
j=3 & F_1^3=\divi{8}\rar\dar[leftarrow, "0.125~"' {blue,at start}]
      & F_2^3=G^3\rar\dar[leftarrow, "0.236~"' {blue,at start}]
        \ular[dotted,"?" description]
      & \red{F_3^3=H^3}\rar\dar[leftarrow, "0.317~"' {blue,at start}]
        \ular[red]
      & F_4^3\rar\dar[leftarrow, "0.380~"' {blue,at start}]
        \ular
      & F_5^3  \dar[leftarrow, "0.430~"' {blue,at start}]
        \ular \\
j=4 & F_1^4=\divi{16}\rar\dar[leftarrow, "0.062~"' {blue,at start}]
      & F_2^4=G^4\rar\dar[leftarrow, "0.145~"' {blue,at start}]
      & F_3^4=H^4\rar\dar[leftarrow, "0.216~"' {blue,at start}]
        \ular[dotted,"?" description]
      & \red{F_4^4}\rar\dar[leftarrow, "0.275~"' {blue,at start}]
        \ular[red]
      & F_5^4  \dar[leftarrow, "0.324~"' {blue,at start}]
        \ular\\
j=5 & F_1^5=\divi{32}\rar\dar[white, "0.031~"' {blue,at start}]
      & F_2^5=G^5\rar\dar[white, "0.090~"' {blue,at start}]
      & F_3^5=H^5\rar\dar[white, "0.147~"' {blue,at start}]
      & F_4^5\rar\dar[white, "0.199~"' {blue,at start}]
        \ular[dotted,"?" description]
      & \red{F_5^5} \dar[white, "0.245~"' {blue,at start}]
        \ular[red] \\[-0.5cm]
    & { } & { } & { } & { } & { }
\end{tikzcd}
\caption{The known $(F_k^j,\le)$ lattice, displayed here for $1\le j,k\le 5$.}
\label{f:Fkj-lattice}
\end{figure}

In Figure~\ref{f:Fkj-lattice}, some dotted edges with question
marks indicate conjectured inequalities.
Indeed, we conjecture that $F_k^{k+1} \ge F_{k+1}^{k+2}$ for all
$k\ge 1$, i.e., that Theorem~\ref{t:2j} may be extended to the case
$j=k+1$. Actually,
Corollary~\ref{Fk_lt_FSk_eventually} already proved that
$F_k^{k+1}(n) > F_{k+1}^{k+2}(n)$ for sufficiently large~$n$.
Thanks to Proposition~\ref{galoisbis}, this conjecture can also be
equivalently formulated as $L_{k+1}^{k+2} \ge L_{k}^{k+1}$. Said
otherwise, the inequality \eqref{e:tauj} appears to still hold as a
non-strict inequality in the case $j=k+1$. Note in this case that we may
indeed reach equality: we can
prove $L_k^{k+1}(k+1) = L_{k+1}^{k+2}(k+1)$.

After that, for $j\ge k+2$, we prove below that
$F_k^j \ngeq F_{k+1}^{j+1}$. More precisely, when $j\ge k+2$,
two behaviors seem possible: either
$k+2\le j\le 2k$, and we prove below that
$F_k^j$ and $F_{k+1}^{j+1}$ are incomparable in this case;
or $j > 2k$, in which case we conjecture that $F_k^j \le F_{k+1}^{j+1}$.
For studying these questions, we now focus on
$L_{k+1}^{j+1}(1)-L_k^j(1)$.

\begin{lemma} \label{l:Lkj1}
For all $k \ge 1$, we have
\[
L_{k+1}^{j+1}(1)-L_k^j(1) =
\begin{cases}1, &\mbox{if}\ 0 \le j \le 2k; \\[.5ex] - \frac{(j-2k-1)(j-2k+2)}{2}, & \mbox{if}\ 2k \le j \le 3k.\end{cases}
\]
In particular $L_{k+1}^{2k+2}(1) = L_{k}^{2k+1}(1)$.
Moreover for all $j \ge 2k+2$, we have $L_{k+1}^{j+1}(1) < L_k^j(1)$.
\end{lemma}
\begin{proof}
From Proposition~\ref{basicL}, we have
\[
L_k^j(1) = \begin{cases}j+1, & \mbox{if}\ 0 \le j \le k; \\[.5ex] L_k^{j-1}(1) + L_k^{j-k}(1), & \mbox{if}\ j \ge k.\end{cases}
\]

For $0 \le j \le k$, we have $L_{k+1}^{j}(1) = j+1 = L_k^j(1)$ and
$L_{k+1}^{j+1}(1) = j+2 = L_k^j(1) + 1$.
For $k < j \le 2k$, we have
\[
\begin{aligned}
L_{k+1}^{j+1}(1) - L_k^j(1) & = L_{k+1}^j(1) - L_k^{j-1}(1) = 1, \\
L_{k+1}^j(1) - L_k^j(1) & = L_{k+1}^{j-1}(1) - L_k^{j-1}(1) - 1 = k-j,
\end{aligned}
\]
by induction on~$j$.
Note that $ -  \frac{(j-2k-1)(j-2k+2)}{2} = 1$ for $j = 2k$.
For $2k < j \le 3k$, we have~thus
\[
L_{k+1}^{j+1}(1) - L_k^j(1)
= - \frac{(j - 2k - 2)(j - 2k + 1)}{2} + 2k - j
= - \frac{(j - 2k - 1)(j - 2k + 2)}{2},
\]
by induction on~$j$.

Still for $k\ge 1$, we now prove by induction on~$j$ that
$L_{k+1}^{j+1}(1) < L_k^j(1)$ for all $j \ge 2k+2$.
This is true for $j=2k+2$: indeed, either $k=1$ and we directly
compute $L_2^5(1) = 13 < 16 = L_1^4(1)$, or $k\ge 2$ and hence
$2k+2 = j \le 3k$
so $L_{k+1}^{j+1}(1) - L_k^j(1)$ is given by the previous
formula, which is strictly negative here. Now, for the step case of the
induction, let $j > 2k+2$. Then
\[ L_{k+1}^{j+1}(1) - L_k^j(1)
   = (L_{k+1}^{j}(1)-L_k^{j-1}(1)) + (L_{k+1}^{j-k}(1)-L_k^{j-k}(1)).
\]
The induction hypothesis on~$j-1$ indicates that the central difference
above is strictly negative, while the rightmost difference is
nonpositive thanks to Corollary~\ref{c:1bis}, allowing us to conclude
this induction.
\end{proof}

\begin{lemma}\label{l:Lkjnj}
Let $k\ge 1$. If $j\ge k+2$, the following value $n_j$
satisfies $L_k^j(n_j) > L_{k+1}^{j+1}(n_j)$:
\[
n_j = \begin{cases}2k+3-j, & \mbox{if}\ \ k+2 \le j \le 2k+2; \\[.5ex]
1, & \mbox{if}\ \ 2k+2 \le j.  \end{cases}
\]
\end{lemma}
\begin{proof}
The case $j\ge 2k+2$ where $n_j=1$ is a direct use of
Lemma~\ref{l:Lkj1}.
Now suppose $k+2 \le j \le 2k+2$.
By Proposition~\ref{basicL}, we know that $L_k^p(1) = p+1$ for all
$0\le p \le k$. By considering $p=n_j-1$, we obtain
$n_j = L_k^{n_j-1}(1)$ and hence
\[
L_k^j(n_j) = L_k^j(L_k^{n_j-1}(1)) = L_k^{j+n_j-1}(1) = L_k^{2k+2}(1).
\]
Similarly, $n_j = L_{k+1}^{n_j-1}(1)$ and
$L_{k+1}^{j+1}(n_j) = L_{k+1}^{2k+3}(1)$, which is strictly less than
$L_k^{2k+2}(1)$ by Lemma~\ref{l:Lkj1}.
\end{proof}

In particular, this last lemma implies that
$L_{k}^j \nleq L_{k+1}^{j+1}$
when $j \ge k+2$, with $n_j$ as counterexample.
Thanks to Proposition~\ref{galoisbis},
this means equivalently that $F_k^j \ngeq F_{k+1}^{j+1}$ when
$j \ge k+2$,
with $n=L_k^j(n_j)$ as counterexample. Moreover, when
$k+2\le j\le 2k$, Lemma~\ref{l:Lkj1} implies that
$L_k^j \ngeq L_{k+1}^{j+1}$, with $n=1$ as counterexample, and
hence that $L_k^j$ and $L_{k+1}^{j+1}$ are incomparable, and
equivalently that $F_k^j$ and $F_{k+1}^{j+1}$ are also incomparable.

We conclude this section with a last conjecture about~$F_k$, as always
for $k\ge 1$.
From Corollary~\ref{Fk_lt_FSk_eventually}, we know that
the inequality $F_k(n) \le F_{k+1}(n)$ becomes strict when
$n$ is large enough.
Actually, we conjecture an explicit bound
$N_k = \tfrac{1}{2}(k+1)(k+6)$, for which
$F_k(n) < F_{k+1}(n)$ as soon as $n>N_k$.
At least, it can be proved that $F_k(N_k) = F_{k+1}(N_k)$, so the bound
cannot be less than~$N_k$, but it
remains to be confirmed that no equality occurs after~$N_k$.
These constants~$N_k$
can also be expressed as $\tfrac{1}{2}(k+3)(k+4) - 3$ and satisfy
$N_{k+1} = N_k + (k+4)$. In particular $N_1 = 7$ and $N_2 = 12$ and
$N_3 = 18$.
Interestingly, we also have $L_{k+1}(N_k) = L_{k+2}(N_k)$.
Finally, this conjecture implies two other interesting statements:
\begin{itemize}
\item $F_k(n) < F_{k+1}(n+1)$ for all $n \ge 2$;
\item $L_{k+1}(n) > L_{k+2}(n)$ for all $n > N_k$.
\end{itemize}

\section{More letter counting}
\label{s:morecount}

Several results and conjectures of the previous section can be rephrased
into statements about the number of occurrences of certain letters in words~$x_k$. In particular,
let us consider again the letter~$1$ and study~$C_k^{(=1)}$.

\begin{proposition}\label{p:C1decr}
For all $k\ge 1$ and $n\ge 0$, we have
$C_k^{(=1)}(n) \ge C_{k+1}^{(=1)}(n)$.
\end{proposition}

\begin{proof}
When $k \ge 2$, this is a consequence of
Proposition~\ref{FC1} and Theorem~\ref{t:2}. For $k=1$, we have $C_1^{(=1)}(n)= n \ge C_2^{(=1)}(n)$.
\end{proof}

Considering the letter~$2$, we conjecture that
$C_k^{(=2)}(n) \ge C_{k+1}^{(=2)}(n)$ for all $k \ge 3$, $n\ge 0$.
In particular, this is a consequence of the conjecture
$F_k^{k+1} \ge F_{k+1}^{k+2}$ mentioned in the previous section:
when combining it with Equation~\eqref{e:FCi}, we get
\[
C_k^{(=2)}(n) \ge C_{k+1}^{(=2)}(n) \quad \mbox{for all}\ k>2
\ \mbox{and}\ n\ge 2;
\]
we can then relax the condition on $n$ because all these quantities
are zero when $n$ is $0$ or~$1$. (In
the latter case, $x_k[0] = k \neq 2$ and similarly for~$x_{k+1}$.)

Actually, the property $C_k^{(=2)} \ge C_{k+1}^{(=2)}$ also
extends to $k=2$, and is easy to prove
in this case, since for all $n\ge 0$, we have
$C_2^{(=2)}(n) \ge C_3^{(=3)}(n)$
(by Equation~\ref{e:tau}, trivially extended to $n=0$) as well as
$C_3^{(=3)}(n) \ge C_3^{(=2)}(n)$; indeed, in~$x_3$, any occurrence of the letter~$2$ is in a subword~$312$.
We cannot extend further: for $k=1$, there is no letter~$2$ in the word~$x_1$, hence $C_1^{(=2)}(n) = 0$ while $C_2^{(=2)}(n) = F_2(n)$ by Equation~\ref{e:FCk}, which differs from~$0$ as soon as $n \ge 1$.

For the letters $3\le i < k$, there is no pointwise monotonicity
between $C_k^{(=i)}$ and $C_{k+1}^{(=i)}$:

\begin{proposition}
For $3\le i < k$, we have
\begin{itemize}
\itemsep.5ex
\item $C_k^{(=i)}(n) < C_{k+1}^{(=i)}(n)$ when $n = i + L_k^{2k+2}(1)$,
\item $C_k^{(=i)}(n) > C_{k+1}^{(=i)}(n)$ when $n = i + L_{k+1}^{k+i}(1)$.
\end{itemize}
\end{proposition}

\begin{proof}
After Lemma~\ref{l:Lkjnj}, we noticed that $F_k^j$ and~$F_{k+1}^{j+1}$
are incomparable for all $j\ge k+2$. Some
counterexamples are $L_k^j(n_j)$ in one direction and $L_{k+1}^{j+1}(1)$
in the other (by Proposition~\ref{galoisbis}).
Now we use Equation~\eqref{e:FCi} to express this in terms of the number of occurrences of letters, by choosing $j = k+i-1$. When $3\le i<k$, we indeed have
$j \ge k+2$ (and also $j \le 2k + 2$).
Due to the shape of Equation~\eqref{e:FCi}, the
previous counterexamples are now shifted by $i$. Moreover, here
\[
n_j = 2k+3-j = k+4-i = L_k^{k+3-i}(1),
\]
and hence
\[
L_k^j(n_j) = L_k^{k+i-1}(L_k^{k+3-i}(1)) = L_k^{2k+2}(1). \qedhere
\]
\end{proof}

For instance, for $k=5$ and $i=4$, we have
$C_5^{(=4)}(49) = 5 < 6 = C_6^{(=4)}(49)$ while
$C_5^{(=4)}(20) = 2 > 1 = C_6^{(=4)}(20)$.

Finally, we compare these quantities for sufficiently large~$n$.
Let $k\ge 1$ and $1\le i < k$.
By Corollary~\ref{Fk_lt_FSk_eventually} and
Equation~\eqref{e:FCi}, we obtain that
$C_k^{(=i)} >_\infty C_{k+1}^{(=i)}$ at least when
$\alpha_k \ge 1 - \frac{1}{i}$, which in particular happens when
$(i-1)^2 \le k$. Note that this condition is always satisfied
when $i=1$ or $i=2$. Otherwise, for $i\ge 3$, small values
of~$k$ may exhibit the opposite infinitary behavior. For example
$C_6^{(=5)} <_\infty C_{7}^{(=5)}$ since
\[
\lim_{n\to\infty} \tfrac{1}{n} C_6^{(=5)}(n) = \alpha_6^{10} \approx 0.0813, \quad \lim_{n\to\infty} \tfrac{1}{n} C_7^{(=5)}(n) = \alpha_7^{11} \approx 0.0819.
\]

\section{Acknowledgments}
The authors are deeply grateful to Yining Hu, who made this joint work
possible. We also thank the anonymous referees for their helpful remarks.

The third author was supported by the ERC grant DynAMiCs (101167561) of the European Research Council, the bilateral grant SYMDYNAR (ANR-23-CE40-0024 and FWF I 6750) of the Agence Nationale de la Recherche and the Austrian Science Fund, and by the ANR project IZES (ANR-22-CE40-0011).

\bigskip
\hrule
\bigskip

\noindent 2020 {\it Mathematics Subject Classification}:
Primary 11B37; Secondary 11B39, 68R15, 68V15.

\noindent \emph{Keywords: } Hofstadter sequence, Fibonacci word, theorem prover, proof assistant, Rocq.

\bigskip
\hrule
\bigskip

\noindent (Concerned with sequences
\seqnum{A005206},
\seqnum{A005374},
\seqnum{A005375},
\seqnum{A005376}, and
\seqnum{A100721}.)

\bigskip
\hrule
\bigskip

\vspace*{+.1in}
\noindent
Received September 27 2024; 
revised versions  received April 17 2026; May 15 2026.
Published in {\it Journal of Integer Sequences}, May 19 2026.

\bigskip
\hrule
\bigskip

\noindent
Return to \href{https://cs.uwaterloo.ca/journals/JIS/}{Journal of Integer Sequences home page}.
\vskip .1in

\end{document}